\documentclass[showpacs,preprintnumbers,pra,aps,amssymb,amsfonts]{revtex4}

\usepackage{times,graphicx}

\usepackage{amssymb,amsmath}
\usepackage{amsthm}
\usepackage{bbm} 
\usepackage{amsfonts,mathrsfs}
\usepackage{url}
\usepackage{enumerate}
\usepackage{txfonts}

\def\ot{\otimes}
\def\H{\textsf{H}}\def\S{\textsf{S}}

\newcommand{\out}[2]{| #1\rangle\langle #2 |}

\newcommand{\pa}[1]{(#1)}
\newcommand{\Pa}[1]{\left(#1\right)}

\newcommand{\set}[1]{\{#1\}}
\newcommand{\Set}[1]{\left\{#1\right\}}

\newcommand{\Bra}[1]{\left\langle#1\right|}
\newcommand{\ket}[1]{|#1\rangle}
\newcommand{\Ket}[1]{\left|#1\right\rangle}

\DeclareMathOperator{\trace}{Tr}
\newcommand{\ptr}[2]{\trace_{#1}\pa{#2}}
\newcommand{\Ptr}[2]{\trace_{#1}\Pa{#2}}

\newcommand{\Tr}[1]{\Ptr{}{#1}}

\newcommand{\identity}{\mathbbm{1}}

\def\cH{\mathcal{H}}

\def\sH{\mathscr{H}}

\newtheorem{thrm}{Theorem}[section]

\theoremstyle{definition}

\newtheorem{exam}[thrm]{Example}
\numberwithin{equation}{section}

\begin{document}


\title{On Conjectures of Classical and Quantum Correlations in Bipartite States}
\author{Lin Zhang and Junde Wu}
\affiliation{Department of Mathematics, Zhejiang University,
Hangzhou, 310027, P.~R.~China}
\affiliation{E-mail: {godyalin@163.com, linyz@zju.edu.cn}}

\begin{abstract}
In this paper, two conjectures which were proposed in [Phys.~Rev.~A
\textbf{82}, 052122(2010)] on the correlations in a bipartite state
$\rho^{AB}$ are studied. If the mutual information $I\Pa{\rho^{AB}}$
between two quantum systems $A$ and $B$ before any measurement is
considered as the total amount of correlations in the state
$\rho^{AB}$, then it can be separated into two parts: classical
correlations and quantum correlations. The so-called classical
correlations $C\Pa{\rho^{AB}}$ in the state $\rho^{AB}$, defined by
the maximizing mutual information between two quantum systems $A$
and $B$ after von Neumann measurements on system $B$, we show that
it is upper bounded by the von Neumann entropies of both subsystems
$A$ and $B$, this answered the conjecture on the classical
correlation. If the quantum correlations $Q\Pa{\rho^{AB}}$ in the
state $\rho^{AB}$ is defined by $Q\Pa{\rho^{AB}}= I\Pa{\rho^{AB}} -
C\Pa{\rho^{AB}}$, we show also that it is upper bounded by the von
Neumann entropy of subsystem $B$. It is also obtained that
$Q\Pa{\rho^{AB}}$ is upper bounded by the von Neumann entropy of
subsystem $A$ for a class of states.

\end{abstract}
\maketitle

\section{Introduction}

In quantum information theory, each realizable physical set-up that
processes states of quantum system is described by a quantum
operation \cite{Kraus} which is mathematically represented by a
linear, completely positive super-operator from a set of quantum
states to another. The information encoded in a given quantum state
is quantified by its von Neumann entropy. In general, the
decoherence will be induced in the quantum system when the quantum
state is acted by a quantum operation. There are few general and
quantitative investigation on the decorrelating capabilities of
quantum operations although the decoherent effects of quantum
operations are popularly realized.

In order to investigate the decorrelating capabilities of quantum
operations, Luo \cite{Luo} suggested that the decorrelating
capabilities of quantum operations should be separated into
classical and quantum parts, and the decoherence involved should be
related to the quantum part. By the duality of quantum operations
and quantum states, each quantum operation can be identified with a
bipartite state via the well-known Choi-Jamio{\l}kowski isomorphism
\cite{Choi}. Thus the study of the decorrelating capabilities of
quantum operations may be transformed into the investigation of
correlations of its corresponding Choi-Jamio{\l}kowski bipartite
states. In view of this, the total correlations in a bipartite state
play an essential role in the study of the decorrelating
capabilities of quantum operations. In order to get some finer
quantitative results, after the total correlation was separated into
classical and quantum parts, two related conjectures were proposed in \cite{Luo} with some supporting examples. In this paper,
the two conjectures are investigated.

\section{Classical and quantum correlations in bipartite states}\label{sec:correlation}

Let $\sH^1$ be a finite dimensional complex Hilbert space. A quantum
operation $\Phi$ on $\sH^1$ is a completely positive linear
super-operator defined on the set of the quantum states on $\sH^1$.
It follows from (\cite{Watrous}, Prop. 5.2 and Cor. 5.5) that there
exists linear operators $\{M_\mu\}_{\mu=1}^K$ on $\sH^1$ such that
$\sum_{\mu=1}^K M^\dagger_\mu M_\mu = \identity^1$ and for each
quantum state $\rho$ on $\sH^1$, we have the Kraus representation
\begin{eqnarray*}
\Phi(\rho) = \sum_{\mu=1}^K M_\mu \rho M^\dagger_\mu.
\end{eqnarray*}
Moreover, let $\sH^2=\mathbb{C}^K$ and $\set{\ket{\mu}}_{\mu=1}^K$
be the standard orthonormal basis of $\sH^2$. If we define $V: \sH^1
\longrightarrow \sH^1 \ot \sH^2$ by \begin{eqnarray*} V\ket{\psi} &
= & \sum_{\mu=1}^K M_\mu \ket{\psi} \ot \ket{\mu},\quad \ket{\psi}
\in \sH^1,
\end{eqnarray*}
then $V$ is an isometry and for each quantum state $\rho$ on
$\sH^1$, we have the Stinespring representation
\begin{eqnarray*}
\Phi(\rho) = \Ptr{2}{V\rho V^\dagger}.
\end{eqnarray*}
It is easy to see that \begin{eqnarray*} V\rho V^\dagger & = &
\sum_{\mu,\nu} M_\mu \rho M^\dagger_\nu \ot \out{\mu}{\nu}.
\end{eqnarray*}
On the other hand, note that for each state $\rho$ on $\sH^1$,
$\ptr{1}{V \rho V^\dagger}$ is a state on $\sH^2$, thus, the map
\begin{eqnarray*}
\widehat{\Phi} : \rho \mapsto \Ptr{1}{V \rho V^\dagger} =
\sum_{\mu,\nu} \Tr{M_\mu \rho M^\dagger_\nu} \out{\mu}{\nu}
\end{eqnarray*}
is a quantum operation from quantum system $\sH^1$ to quantum system
$\sH^2$, we call it \emph{complementary} to $\Phi$.

If we consider $\sH^2$ to be the environment, then
$\widehat{\Phi}(\rho)$ is the state of the environment after the
interaction and is called a \emph{correlation matrix}. If the
initial state $\rho$ is pure, then the von Neumann entropy
$$
\S\Pa{\widehat{\Phi}(\rho)} = -
\Tr{\widehat{\Phi}(\rho)\log_2\widehat{\Phi}(\rho)}
$$
of
$\widehat{\Phi}(\rho)$ describes the entropy exchanged between
the system and the environment. Therefore,
$\S\Pa{\widehat{\Phi}(\rho)}$ is called the \emph{exchange entropy}.
The relationship among the $\S(\Phi(\rho)), \S(\rho)$, and
$\S\Pa{\widehat{\Phi}(\rho)}$ is connected by the well-known
Lindblad's entropy inequality \cite{Lindblad}:
$$
\left|\S\Pa{\widehat{\Phi}(\rho)} - \S(\rho)\right| \leqslant \S(\Phi(\rho))
\leqslant \S\Pa{\widehat{\Phi}(\rho)} + \S(\rho). \eqno (1)
$$
It follows from $\sum_{\mu=1}^K M^\dagger_\mu M_\mu = \identity^1$
that $\set{ M_\mu}_{\mu=1}^K$ describes a measurement which
transforms the initial state $\rho$ into one of the output states
$$
\rho'_\mu = \frac{1}{q_\mu} M_\mu \rho M^\dagger_\mu
$$
with probability $q_\mu = \Tr{M_\mu \rho M^\dagger_\mu}$. Thus,
$\Set{q_\mu,\rho'_\mu}$ is a quantum ensemble and its \emph{Holevo
quantity} is defined by
\begin{eqnarray*}
\chi\Pa{\Set{q_\mu,\rho'_\mu}} = \S\Pa{\sum_\mu q_\mu \rho'_\mu} - \sum_\mu q_\mu \S\Pa{\rho'_\mu}.
\end{eqnarray*}
Let $\H\Pa{\set{q_\mu}} = - \sum_{\mu=1}^Kq_\mu\log_2q_\mu$ be the
Shannon entropy of the probability distribution $\set{q_\mu}$. Then
we have the following inequality \cite{Roga1}:
$$\chi\Pa{\Set{q_\mu,\rho'_\mu}} \leqslant \S\Pa{\widehat{\Phi}(\rho)}
\leqslant \H\Pa{\set{q_\mu}}. \eqno (2)
$$

Let $\sH^R$ and $\sH^Q$ be two finite dimensional complex Hilbert
spaces. If $\Phi^Q$ is a quantum operation on $\sH^Q$, then
$\identity^R \ot \Phi^Q$ is a quantum operation on $\sH^R\ot \sH^Q$,
moreover, if $\rho^{RQ}$ is a state on $\sH^R\ot \sH^Q$ and $\rho^Q
= \Ptr{R}{\rho^{RQ}}$, then we have \cite{Roga2}:
$$
\S\Pa{\widehat{\identity^R \ot \Phi^Q}\Pa{\rho^{RQ}}} =
\S\Pa{\widehat{\Phi}^Q\Pa{\rho^Q}}. \eqno (3)
$$
Let $\sH^A$ and $\sH^B$ be two finite dimensional complex Hilbert
spaces, $\rho^{AB}$ is a state on $\sH^A\ot \sH^B$, $\rho^A =
\Ptr{B}{\rho^{AB}}$, $\rho^B = \Ptr{A}{\rho^{AB}}$. Then the total
correlations in $\rho^{AB}$ are usually quantified by the quantum
mutual information
\begin{eqnarray*}
I\Pa{\rho^{AB}} = \S\Pa{\rho^A} + \S\Pa{\rho^B} - \S\Pa{\rho^{AB}}.
\end{eqnarray*}

In \cite{Luo}, the author separated the total correlations
$I\Pa{\rho^{AB}}$ into classical correlations $C\Pa{\rho^{AB}}$ and quantum
correlations \cite{Ollivier} $Q\Pa{\rho^{AB}} = I\Pa{\rho^{AB}} - C\Pa{\rho^{AB}}$,  where $C\Pa{\rho^{AB}}$ was defined by
$$
C\Pa{\rho^{AB}} = \sup_{\Pi^B}I\Pa{\Pi^B\Pa{\rho^{AB}}},
$$
the sup is taken over all von Neumann measurements $\Pi^B =
\Set{\Pi^B_j}$ on $\sH^B$, and
\begin{eqnarray*}
\Pi^B\Pa{\rho^{AB}} = \sum_j \Pa{\identity^A \ot \Pi^B_j}\rho^{AB}\Pa{\identity^A \ot \Pi^B_j}
\end{eqnarray*}
is the output state after executing the nonselective measurement
$\Pi^B = \Set{\Pi^B_j}$; $\identity^A$ is the identity operator on
$\sH^A$.

In \cite{Luo}, the following conjectures are proposed :
\begin{eqnarray*}
C\Pa{\rho^{AB}} & \leqslant & \min \Set{\S\Pa{\rho^A}, \S\Pa{\rho^B}},\,\,\,\,\,\,\,\,\,\,\,\, (I)\\
Q\Pa{\rho^{AB}} & \leqslant & \min \Set{\S\Pa{\rho^A}, \S\Pa{\rho^B}}.\,\,\,\,\,\,\,\,\,\,\,\, (II)
\end{eqnarray*}

In this paper, we prove the conjecture (I). Moreover, we
show that $Q\Pa{\rho^{AB}} \leqslant \S\Pa{\rho^B}$ is always valid, and
the conjecture (II) is true if $\S\Pa{\rho^B} \leqslant \S\Pa{\rho^A}$ or
$\rho^{AB}$ is separable. It is obtained that
$Q\Pa{\rho^{AB}}$ is upper bounded by the von Neumann entropy of subsystem $A$ for a class of states. So the conjecture (II) is also
true for these states.

\section{The proof of the conjecture}\label{sec:mainresult}

Our main results are the following:

\begin{thrm}\label{th:main}
Let $\rho^{AB}$ be a quantum state on $\sH^A \ot \sH^B$. Then we
have \begin{enumerate}[(i)]
\item $C\Pa{\rho^{AB}} \leqslant \min \Set{\S\Pa{\rho^A}, \S\Pa{\rho^B}}$,

\item $Q\Pa{\rho^{AB}} \leqslant
\S\Pa{\rho^B}$, and $Q(\rho^{AB}) \leqslant \min \set{\S\Pa{\rho^A},
\S\Pa{\rho^B}}$ whenever $\S\Pa{\rho^B} \leqslant \S\Pa{\rho^A}$ or
$\rho^{AB}$ is separable.
\end{enumerate}
\end{thrm}

\begin{proof}
\begin{enumerate}[(i)]
\item  Let $\Set{\Ket{\psi^B_j}}_{j=1}^k$ be a orthonormal basis of
$\sH^B$ and $\Pi^B_j = \Ket{\psi^B_j}\Bra{\psi^B_j}$. Then
$\Tr{\identity^A \ot \Pi^B_j \rho^{AB} \identity^A \ot \Pi^B_j}
= \Bra{\psi^B_j} \rho^B \Ket{\psi^B_j}$. If we denote
$\Bra{\psi^B_j} \rho^B \Ket{\psi^B_j}$ by $p_j$, then  $p_j\geqslant
0$ and $\sum_jp_j=1$. Without loss of generality, we assume that all
$p_j>0$. Now, we define
$$
\rho^A_j = \frac{
\Bra{\psi^B_j} \rho^{AB} \Ket{\psi^B_j}}{p_j},
$$
then $\rho^A_j$ is a state on $\sH^A$ and
$$
\Pi^B\Pa{\rho^{AB}}  = \sum_j p_j \rho^A_j \ot \Pi^B_j, \quad \Pi^B\Pa{\rho^B}  =  \sum_j \Pi^B_j\rho^B\Pi^B_j  =  \sum_j p_j
\Pi^B_j,\quad
\rho^A  = \sum_j p_j \rho^A_j.
$$
Thus,
$$
\S\Pa{\Pi^B\Pa{\rho^{AB}}}  =  \H\Pa{\set{p_j}} + \sum_j p_j \S\Pa{\rho^A_j},\quad \S\Pa{\Pi^B\Pa{\rho^B}}  =  \H\Pa{\set{p_j}},
$$
and
\begin{eqnarray*}
I\Pa{\Pi^B\Pa{\rho^{AB}}} & = & \S\Pa{\rho^A} + \S\Pa{\Pi^B\Pa{\rho^B}} - \S\Pa{\Pi^B\Pa{\rho^{AB}}} \\
 & = & \S\Pa{\rho^A} - \sum_j p_j \S\Pa{\rho^A_j} = \chi\Pa{\Set{p_j,\rho^A_j}}.
\end{eqnarray*}
Note that $\sum_j p_j \S\Pa{\rho^A_j} \geqslant 0$. Hence
$I\Pa{\Pi^B\Pa{\rho^{AB}}}  \leqslant \S\Pa{\rho^A}$. Thus $C\Pa{\rho^{AB}} =
\sup_{\Pi^B}I\Pa{\Pi^B\Pa{\rho^{AB}}}  \leqslant \S\Pa{\rho^A}$. On the other hand, it follows from $C\Pa{\rho^{AB}} =
\sup_{\Pi^B}I\Pa{\Pi^B\Pa{\rho^{AB}}}$ and
$I\Pa{\Pi^B\Pa{\rho^{AB}}}  = \chi\Pa{\Set{p_j,\rho^A_j}}$ that in order to prove
$C\Pa{\rho^{AB}} \leqslant \S\Pa{\rho^B}$, we only need to prove
$\chi\Pa{\Set{p_j,\rho^A_j}} \leqslant \S\Pa{\rho^B}$. Note that the
quantum ensemble $\Set{p_j,\rho^A_j}$ is obtained from the quantum
operation of taking partial trace over $\cH^B$ from the quantum
state $\rho^{AB}$, this inspired us to define the following quantum
operation $\Psi$ on the quantum system $\sH^A \ot \sH^B$: Let $\Ket{\omega^B} \in \sH^B$ be a fixed unit vector, for each
quantum state $\sigma^{AB}$ on $\sH^A \ot \sH^B$,
\begin{eqnarray*}
\Psi\Pa{\sigma^{AB}} & = & \sum_{j} \Pa{\identity^A \ot \Ket{\omega^B}\Bra{\psi^B_j}} \sigma^{AB} \Pa{\identity^A \ot \Ket{\omega^B}\Bra{\psi^B_j}} \\
& = & \Ptr{B}{\sigma^{AB}} \ot \Ket{\omega^B}\Bra{\omega^B}.
\end{eqnarray*}

Let $\sH^C=\mathbb{C}^k$ and $\set{\ket{i}}_{i=1}^k$ be the standard
orthonormal basis of $\sH^C$. Then the correlation matrix
$\widehat{\Psi}\Pa{\rho^{AB}}$ is given by
\begin{eqnarray*}
\widehat{\Psi}\Pa{\rho^{AB}} & = & \sum_{i,j} \Tr{\identity^A \ot \Ket{\omega^B}\Bra{\psi^B_i} \rho^{AB} \identity^A \ot \Ket{\psi^B_j}\Bra{\omega^B}}\out{i}{j}\\
& = & \sum_{i,j}\Bra{\psi^B_i}\rho^B\Ket{\psi^B_j}\out{i}{j},
\end{eqnarray*}

If we define $W = \sum_j\Ket{j}\Bra{\psi^B_j}$, then $W^\dagger W =
\identity^B, WW^\dagger = \identity^C$, that is, $W$ is an unitary
operator from $\cH^B$ to $\cH^C$. It follows from
$\widehat{\Psi}\Pa{\rho^{AB}} = W \rho^B W^\dagger$ that
$\S\Pa{\widehat{\Psi}\Pa{\rho^{AB}}} = \S\Pa{\rho^B}$. Note that the
quantum ensemble $\Set{p_j,\rho^A_j \ot \Ket{\omega^B}\Bra{\omega^B}}$
can be obtained by the quantum operation $\Psi$ and
$\chi\Pa{\Set{p_j,\rho^A_j}} = \chi\Pa{\Set{p_j,\rho^A_j \ot
\Ket{\omega^B}\Bra{\omega^B}}}$. By using the inequality (2) we have
\begin{eqnarray*}
\chi\Pa{\Set{p_j,\rho^A_j}} &=& \chi\Pa{\Set{p_j,\rho^A_j \ot
\Ket{\omega^B}\Bra{\omega^B}}} \leqslant \S\Pa{\widehat{\Psi}\Pa{\rho^{AB}}} = \S\Pa{\rho^B}.
\end{eqnarray*}
Thus, we have proved $C\Pa{\rho^{AB}} \leqslant \min \Set{\S\Pa{\rho^A}, \S\Pa{\rho^B}}$.

\item Note that equality (3) shows that
$\S\Pa{\widehat{\identity^A \ot \Pi^B}\Pa{\rho^{AB}}} =
\S\Pa{\widehat{\Pi}^B\Pa{\rho^B}}$. Hence it follows from inequality
(1) that
$$
\S\Pa{\Pi^B\Pa{\rho^{AB}}} - \S\Pa{\rho^{AB}} \leqslant
\S\Pa{\widehat{\identity^A \ot \Pi^B}\Pa{\rho^{AB}}} =
\S\Pa{\widehat{\Pi}^B\Pa{\rho^B}} =
\H(\set{p_j}) = \S\Pa{\Pi^B\Pa{\rho^B}}.\eqno (4)
$$
On the other hand,
note that
$I\Pa{\Pi^B\Pa{\rho^{AB}}} = \S\Pa{\rho^{A}} + \S\Pa{\Pi^B\Pa{\rho^{B}}} - \S\Pa{\Pi^B\Pa{\rho^{AB}}}$,
by the definition of $Q(\rho^{AB})$ and inequality (4) we have
$$
Q\Pa{\rho^{AB}} = I\Pa{\rho^{AB}} - C\Pa{\rho^{AB}} \leqslant \S\Pa{\Pi^B
\rho^{AB}} - \S\Pa{\rho^{AB}} - \S\Pa{\Pi^B \rho^{B}} + \S\Pa{\rho^B}
\leqslant \S\Pa{\rho^B}.
$$
This showed that $Q\Pa{\rho^{AB}} \leqslant \S\Pa{\rho^B}$. Clearly, when $\S\Pa{\rho^B} \leqslant \S\Pa{\rho^A}$, it follows from
$Q\Pa{\rho^{AB}} \leqslant \S\Pa{\rho^B}$ that $Q\Pa{\rho^{AB}} \leqslant
\min \Set{\S\Pa{\rho^A}, \S\Pa{\rho^B}}$. If $\rho^{AB}$ is a separable state, then $\S\Pa{\rho^{AB}} \geqslant
\max \Set{\S\Pa{\rho^A}, \S\Pa{\rho^B}}$ \cite{Fan}. Note that
$I\Pa{\Pi^B\Pa{\rho^{AB}}} \geqslant 0$, so $\S\Pa{\rho^B} - \S\Pa{\rho^{AB}}
\leqslant I\Pa{\Pi^B\Pa{\rho^{AB}}}$. Thus, we can prove easily that
$Q\Pa{\rho^{AB}} \leqslant \min \Set{\S\Pa{\rho^A}, \S\Pa{\rho^B}}$. The
theorem is proved.

\end{enumerate}
\end{proof}
\,\,\,\,

In what follows, in order to provide a class of states $\rho^{AB}$
such that $Q\Pa{\rho^{AB}}
\leqslant \S\Pa{\rho^A}$, we need the following:

 \begin{thrm}\label{th:main'}
Let $\sH^B$ and $\sH^C$ be two finite dimensional complex Hilbert
spaces, $\rho^{BC}$ be a state on $\sH^B\ot \sH^C$, $\rho^B =
\Ptr{C}{\rho^{BC}}$, $\rho^C = \Ptr{B}{\rho^{BC}}$. Then
$\S\Pa{\rho^{BC}} = \S\Pa{\rho^B} - \S\Pa{\rho^C}$ if and only if
\begin{enumerate}[(i)]
\item $H^B$ can be factorized into the form $\sH^B = \sH^L \ot \sH^R$, and
\item $\rho^{BC} = \rho^L \ot \Ket{\Psi^{RC}}\Bra{\Psi^{RC}}$, where $\Ket{\Psi^{RC}} \in \sH^R \ot \sH^C$.
\end{enumerate}
\end{thrm}

\begin{proof}
$(\Longleftarrow)$ It is trivially. \\
$(\Longrightarrow)$ The quantum state $\rho^{BC}$ can be purified into a
tripartite state $\Ket{\Omega^{ABC}} \in \sH^A \ot \sH^B \ot \sH^C$,
where $\sH^A$ is a reference system. If we denote $\rho^{ABC} =
\Ket{\Omega^{ABC}}\Bra{\Omega^{ABC}}$, then
\begin{eqnarray*}
\Ptr{AB}{\rho^{ABC}} = \rho^C, \Ptr{AC}{\rho^{ABC}} = \rho^B,\\
\Ptr{C}{\rho^{ABC}} = \rho^{AB}, \Ptr{A}{\rho^{ABC}} = \rho^{BC}.
\end{eqnarray*}
Note that $\S\Pa{\rho^{ABC}} = 0$, so $\S\Pa{\rho^C} = \S\Pa{\rho^{AB}}$,
thus, we have
$$
\S\Pa{\rho^{AB}} + \S\Pa{\rho^{BC}} = \S\Pa{\rho^B} = \S\Pa{\rho^B} + \S\Pa{\rho^{ABC}},
$$
it follows from \cite{Hayden} that

\begin{enumerate}[(i)]
\item $\sH^B$ can be factorized into the form $\sH^B = \bigoplus_{k=1}^K \sH_k^L \ot \sH_k^R$,
\item $\rho^{ABC} = \bigoplus_{k=1}^K \lambda_k \rho_k^{AL} \ot \rho_k^{RC}$, \,\, where $\rho_k^{AL}$ is a state on $\sH^A \ot \sH_k^L$, $\rho_k^{RC}$ is a state on $\sH_k^R \ot \sH^C$, $\set{\lambda_k}$ is a probability distribution.
\end{enumerate}

That $\S\Pa{\rho^{BC}} = \S\Pa{\rho^B} - \S\Pa{\rho^C}$ implies  $\S\Pa{\rho^A}
+ \S\Pa{\rho^C} = \S\Pa{\rho^{AC}}$ is clear, and $\S\Pa{\rho^A}
+ \S\Pa{\rho^C} = \S\Pa{\rho^{AC}}$ if and only if $\rho^{AC} = \rho^A \ot \rho^C$
holds. By the expression form of $\rho^{ABC} = \bigoplus_{k=1}^K
\lambda_k \rho_k^{AL} \ot \rho_k^{RC}$, we have $\rho^{AC} =
\sum_{k=1}^K \lambda_k \rho_k^A \ot \rho_k^C$. Combining these facts
we have $K=1$, i.e., the statement (i) of the theorem holds. Hence
$\rho^{ABC} = \rho^{AL} \ot \rho^{RC}$, where $\rho^{AL}$ is a state
on $\sH^A \ot \sH^L$ and $\rho^{RC}$ is a state on $\sH^R \ot
\sH^C$, it follows from $\rho^{ABC}$ is pure state that both
$\rho^{AL}$ and $\rho^{RC}$ are also pure states. Therefore
$$\rho^{BC} = \Ptr{A}{\rho^{AL}} \ot \rho^{RC} = \rho^L \ot \Ket{\Psi^{RC}}\Bra{\Psi^{RC}}.$$
The statement (ii) holds and the theorem is proved.
\end{proof}

\emph{Note added.} After the present work is completed and submitted
to the arXiv, Luo \emph{et al.} \cite{Li} inform us of the
inequality (II) being not valid in general and provide a
counter-example while they give another approach to the above
inequality (I). We found also that Giorgi \cite{Giorgi} gives a proof to the inequalities (I) and (II) in the case of two qubits by monogamy of discord for pure states.

\begin{exam}

Let $\rho^{AB}$ be a bipartite state on $\sH^A \ot \sH^B$ such that
$\S\Pa{\rho^{AB}} = \S\Pa{\rho^B} - \S\Pa{\rho^{A}}$. By Theorem
\ref{th:main'}, we have $\rho^{AB} = \Ket{\Phi^{AL}}\Bra{\Phi^{AL}} \ot
\rho^R$ for $\Ket{\Phi^{AL}} \in \sH^A \ot \sH^L$, where $\rho^R$ is
a state on $\sH^R$ and $\sH^B = \sH^L \ot \sH^R$. It is easy to show
that although $\S\Pa{\rho^A} \leqslant \S\Pa{\rho^B}$, but $Q\Pa{\rho^{AB}} =
\S\Pa{\rho^A}$, so the conjecture (II) is true for these states. If $\dim \sH^A = \dim\sH^B = 2$, then $\sH^B$ cannot be factorized. This indicates that $\S\Pa{\rho^{AB}} = \S\Pa{\rho^B} - \S\Pa{\rho^{A}}$ implies that $\rho^{AB}$ is pure state. In other words, if two qubit state $\rho^{AB}$ is not pure, then $\S\Pa{\rho^{AB}} > \left|\S\Pa{\rho^B} - \S\Pa{\rho^{A}}\right|$.

\end{exam}

\section{Concluding Remarks}

The conjectures and our work are based on the assumption that
classical correlations are maximized by von Neumann measurement. In
general, it is not true, see \cite{Giorda,Galve,Yu}.

It follows from Theorem~\ref{th:main} that von Neumann measurement
performed on subsystem $B$ induced the following inequality:
$$
\chi\Pa{\Set{p_j,\rho^A_j}} \leqslant \S\Pa{\rho^B}.
$$

By using the above inequality, we studied a conjecture in \cite{Zhang},
proposed by W.~Roga in \cite{Roga1,Roga2}.


\end{document}